\documentclass[12pt]{amsart}

\usepackage{amssymb}
\usepackage{amsmath}
\usepackage{amsthm}
\usepackage{amscd}
\usepackage{enumerate}
\usepackage[margin=1.1in]{geometry}
\usepackage{hyperref}
\usepackage{appendix}
\usepackage{dsfont,txfonts}

\theoremstyle{plain}
\newtheorem{theorem}{Theorem}[section]

\newtheorem{corollary}[theorem]{Corollary}
\newtheorem{lemma}[theorem]{Lemma}
\newtheorem{definition}[theorem]{Definition}
\newtheorem{assumption}[theorem]{Assumption}

\theoremstyle{remark}

\newtheorem{remark}{Remark}[section]
\newtheorem{example}{Example}[section]

\DeclareMathOperator{\supp}{supp}

\DeclareMathOperator{\dom}{dom}

\newcommand{\DF}{\mathcal{E}}
\newcommand{\domDF}{\mathcal{F}}
\newcommand{\Hil}{\mathcal{H}}
\newcommand{\Hsupp}{\supp_{\Hil}}

\begin{document}

\title{Magnetic fields on resistance spaces}
\author{Michael Hinz$^1$}
\address{$^2$ Department of Mathematics, Bielefeld University, Postfach 100131, 33501 Bielefeld, Germany }
\email{mhinz@math.uni-bielefeld.de}
\thanks{$^1$ Research supported in part by SFB 701 of the German Research Council (DFG)}

\author{Luke Rogers$^2$}
\address{$^2$ Department of Mathematics, University of Connecticut, Storrs, CT 06269-3009 USA}
\email{rogers@math.uconn.edu}
\date{\today}

\begin{abstract}
On a metric measure space $X$ that supports a regular, strongly local resistance form we consider a magnetic energy form that corresponds to the magnetic Laplacian for a particle confined to $X$. We provide sufficient conditions for closability and self-adjointness in terms of geometric conditions on the reference measure without assuming energy dominance.
\tableofcontents
\end{abstract}

\maketitle

\section{Introduction}
We study the magnetic Laplacian for a particle in a metric measure space $(X,\mu)$ that supports a regular, strongly local resistance form.  Roughly speaking, a resistance form is a Dirichlet form for which points have positive capacity, and which is determined by its finite-dimensional traces; the formal definition is in Section~\ref{background}.  On such a space there is a Hilbert space $\Hil$ of $1$-forms and a derivation operator $\partial$ that plays the role of a gradient.  In Section~\ref{sec:1-forms} we use these to define a magnetic operator $\partial+ia$, where $a\in\Hil$ is real-valued, and a magnetic (quadratic) form $\langle (\partial+ia)f,(\partial+ia)g\rangle_{\Hil}$.  Our main result, Theorem~\ref{thm:closability}, gives geometric conditions on the measure $\mu$ that suffice for closability of the magnetic form and consequently for the existence of a self-adjoint magnetic Laplacian $\Delta_{\mu,a}$.  The magnetic form and Laplacian have a gauge invariance property which is established in Section~\ref{sec:gauge}. Examples to which the theory may be applied are in Section~\ref{sec:examples}.  

These results complement earlier results of \cite{HTc} where self-adjointness and gauge invariance had been shown for magnetic Schr\"odinger operators in situations where the energy measures are absolutely continuous with respect to the given reference measure (energy dominance). The novelty in this paper is that for the resistance form case this assumption can be replaced by a uniform lower bound for the measure of balls or by a doubling condition. Therefore the present results apply to resistance forms on fractals even if energy and volume are mutually singular, \cite{BBST, Hino03, Hino05}; this case is not covered by \cite{HTc} and is interesting from a spectral theoretic perspective, \cite{FSh92, KiLa93}.

\subsection*{Acknowledgment}
The authors thank the anonymous referee for careful reading and helpful suggestions.

\section{Resistance forms}\label{background}
Following Kigami~\cite{Ki01,Ki03} we define a resistance form as follows.
\begin{definition}
A resistance form $(\DF,\domDF)$ on a  set $X$ is a pair such that:
\begin{enumerate}[{(RF}1)]
\item\label{RF1} $\domDF$ is a linear space of functions $X\to\mathbb{R}$ containing the constants. $\DF$ is a non-negative definite symmetric quadratic form on $\domDF$ with $\DF(u,u)=0$ if and only if $u$ is constant.
\item\label{RF2} The quotient of $\domDF$ by constants is a Hilbert space with norm $\DF(u,u)^{1/2}$.
\item If $v$ is a function on a finite subset $V\subset X$ there is $u\in\domDF$ so $u\,\bigr|_{V}=v$.
\item For $x,y\in X$
\begin{equation*}
	R(x,y)=\sup\Bigl\{ \frac{(u(x)-u(y))^{2}}{\DF(u,u)}:u\in\domDF, \DF(u,u)>0\Bigr\}<\infty.
\end{equation*}
\item\label{RF5} If $u\in\domDF$ then $\bar{u}=\max(0,\min(1,u(x)))\in\domDF$ and $\DF(\bar{u},\bar{u})\leq\DF(u,u)$.
\end{enumerate}
\end{definition}
We write $\DF(u)=\DF(u,u)$ and do similarly for other bilinear expressions. The main feature of resistance forms is that they are determined by a sequence of traces to finite subsets.
\begin{theorem}[\protect{\cite{Ki01,Ki03}}]\label{thm:propsofresistforms}
Resistance forms have the following properties.
\begin{enumerate}
\item $R(x,y)$ is a metric on $X$. Functions in $\domDF$ extend to the completion of $X$ and $(\DF,\domDF)$ is a resistance form on the completion, so we may assume $X$ is complete.
\item If $V\subset X$ is finite there is a trace $\DF_{V}$ of $\DF$ to $V$, which is a resistance form defined by
\begin{equation*}	
	\DF_{V}(v,v)=\inf\Bigl\{\DF(u,u):u\in\domDF, u\,\bigr|_{V}=v\Bigr\}
	\end{equation*}
in which the infimum is achieved at a unique $u$.  Also, if $V_{1}\subset V_{2}$ then $(\DF_{V_{2}})_{V_{1}}=\DF_{V_{1}}$.
\item If $(X,R)$ is separable and $\{V_{n}\}$ is an increasing (under inclusion) sequence of finite sets such that $\cup_{j}V_{j}$ is $R$-dense in $X$ then $\DF_{V_{n}}$ is non-decreasing and $\DF(u)=\lim_{n}\DF_{V_{n}}(u)$ for all $u\in\domDF$.
\item If $\{V_{n}\}$ is an increasing sequence of sets supporting resistance forms $\DF_{n}$ such that $(\DF_{n+1})_{V_{n}}=\DF_{n}$ for all $n$, then  for $u$ defined on  $V_{\ast}=\cup_{n}V_{n}$ the sequence $\DF_{n}(u)$ is non-decreasing and $\DF(u)=\lim_{n\to\infty}\DF_{n}(u)$ defines a resistance form with domain $\domDF=\{u:\DF(u)<\infty\}$.
\item Functions in $\domDF$ are $1/2$-H\"{o}lder in the resistance metric because from the definition of $R(x,y)$ they satisfy
\begin{equation}\label{E:resistest}
|f(x)-f(y)|\leq R(x,y)\mathcal{E}(f)^ {1/2}, \ \ f\in\domDF.
\end{equation}
\end{enumerate}
\end{theorem}


\begin{definition}\label{defn:rdrs}
$(X,R,\mu)$ will be called a regular doubling resistance space if the following hold.
\begin{enumerate}
\item\label{one} There is a resistance form $(\DF,\domDF)$ on $X$, the metric space $(X,R)$ is separable, connected and locally compact, and the compactly supported functions in $\domDF$ are supremum-norm dense in $C_c(X)$.
\item\label{two} $\mu$ is  a non-atomic $\sigma$-finite Borel regular measure with $0<\mu(B(x,r))<\infty$ for all balls $B(x,r)$.
\item\label{three} $X$ is metrically doubling:  there is $C_{d}$ such that any ball $B(x,2r)$ can be covered by $C_{d}$ balls of radius $r$.
\end{enumerate}
\end{definition}

\begin{remark}\label{R:Dform}\mbox{}
\begin{enumerate}
\item[(i)] Condition (\ref{one}) implies that we consider a regular resistance form $(\mathcal{E},\mathcal{F})$,~\cite[Definition 6.2]{Ki12}. In particular, if $K\subset X$ is compact and $U\supset K$ is a relatively compact open neighborhood of $K$ then there exists a compactly supported function $\chi\in \mathcal{F}$ such that $\supp\chi\in U$, $0\leq\chi\leq 1$ on $X$ and $\chi\equiv 1$ on $K$. This follows from \cite[Theorem 6.3]{Ki12}. To such a function $\chi$ we refer as \emph{cut-off function for $K$ and $U$}.
\item[(ii)] If (\ref{one}) holds and in addition $\mu$ is a measure satisfying (\ref{two}) then the space $\mathcal{C}=\domDF\cap C_{c}(X)$ of compactly supported finite energy functions is dense in $L^{2}(\mu)$ and the closure of $(\DF, \mathcal{C})$ on $L^{2}(\mu)$ is a regular Dirichlet form $(\DF,\widetilde{\domDF})$, see~\cite[Theorem 8.4]{Ki12}.  In this case $\mathcal{C}$ is a form core, meaning that it is dense in $\widetilde{\domDF}$ with respect to the norm $(\DF+\|\cdot\|^{2}_{L_2(X,\mu)})^{1/2}$ and also in $C_{c}(X)$ with respect to the supremum norm $\left\|\cdot\right\|_{\sup}$. These facts do not require the metric doubling property (\ref{three}).  Note that we have introduced the notation $\widetilde{\domDF}$ for the domain of the Dirichlet form, which can differ from the resistance form domain $\domDF$ in the case that $(X,R)$ is non-compact. In the compact case the space $\mathcal{C}=\mathcal{F}$ equals $\widetilde{\mathcal{F}}$ (in the sense of distinguished representatives).
\end{enumerate}
\end{remark}
Henceforth we assume that $X$ is a regular doubling resistance space and that $\DF$ is strongly local.  The latter means that $\DF(u,v)=0$ when $u,v\in\domDF$ and $v$ is constant in each component of a neighborhood of $\supp(u)$.

We record the following easy estimate for later use.  For any open ball $B=B(z,r)$ write
\begin{equation*}
	f_{B}=\frac{1}{\mu(B)} \int_{B} f(y)\,d\mu(y),
	\end{equation*}
 and observe that the resistance estimate~\eqref{E:resistest} implies
\begin{equation}\label{eqn:poincare}
	|f(x)-f_{B}|
	\leq \frac{1}{\mu(B)}\int_{B} |f(x)-f(y)|\, d\mu(y)
	\leq \frac{1}{\mu(B)}\DF(f)^{1/2} \int_{B} R(x,y)^{1/2}\, d\mu(y)
	\leq \DF(f)^{1/2} r^{1/2}.
	\end{equation}
Evidently we could modify it by replacing $f_{B}$ by any value in $[\inf_{B}f,\sup_{B}f]$.

We will also need the extension of $\DF$ to complex-valued functions.  If $f=f_{1}+if_{2}$, $g=g_{1}+ig_{2}$ with both $f_{j}$ and $g_{j}$ real-valued elements of $\domDF$, one can set
\begin{equation*}
	\DF(f,g)=\DF(f_{1},g_{1})-i\DF(f_{1},g_{2}) + i\DF(f_{2},g_{1}) + \DF(f_{2},g_{2}).
	\end{equation*}
It is not difficult to check that this is conjugate symmetric and linear in the first variable. We refer to a form with these properties simply as 'quadratic form'. Moreover, the form $\DF$ is non-negative definite, $\DF(f)=\DF(f_{1})+\DF(f_{2})\geq0$. This is what we mean by $\DF(f)$ for a complex-valued $f$.  Observe that~\eqref{eqn:poincare} is still valid for complex-valued $f$.  In what  follows we repeatedly use the natural complexifications of $\DF$, $\domDF$, $L_2(X,\mu)$ and so on, they will be denoted by the same symbols. See \cite{HTc} for more details.



\section{$1$-forms and Magnetic operator}\label{sec:1-forms}
The core $\mathcal{C}=\domDF\cap C_{c}(X)$ is an algebra, and we recall (see, for example, Section~3.2 of~\cite{FOT94}) that regularity and strong locality ensure that associated to $f, g\in\widetilde{\domDF}$ there is a unique Radon measure $\Gamma(f,g)$ on $(X,R)$, called the energy measure, and satisfying
\begin{equation*}
	\mathcal{E}(fh,g)+\mathcal{E}(gh,f)-\mathcal{E}(fg,h)= 2\int_X h\:d\Gamma(f,g),\ \ h\in \mathcal{C}.
	\end{equation*}
In the case $f=g$ the measure is denoted $\Gamma(f)$ and is non-negative; also $\Gamma(f)(X)=\DF(f)$.  Note that no aspect of this construction depends on the measure $\mu$.

Using the energy measures we  can define a nonnegative bilinear form on $\mathcal{C}\otimes\mathcal{C}$ by setting
\[\left\langle a\otimes b, c\otimes d\right\rangle_\mathcal{H}:=\int_Xbd\:d\Gamma(a,c),\qquad a\otimes b, c\otimes d\in\mathcal{C}\otimes\mathcal{C}\] and extending by linearity.  Factoring out zero norm elements and completing yields a Hilbert space $(\mathcal{H},\left\langle \cdot,\cdot\right\rangle_\mathcal{H})$, referred to as the \emph{space of $1$-forms associated with $\mathcal{E}$}.  Note that $\Hil$ contains elements of the form $a\otimes b$ for any $a\in\widetilde{\domDF}$, $b\in L^{2}(d\Gamma(a))$.  In the context of Dirichlet and resistance forms this construction was introduced in \cite{CS03, CS09, IRT} and studied further in \cite{HRT, HT, HTc}.  It is important that $(ab)\otimes c-a\otimes bc-b\otimes ac=0$, see the proof of Theorem~2.7 of~\cite{IRT}.

The algebra $\mathcal{C}$ acts on $\mathcal{C}\otimes\mathcal{C}$ by
\begin{equation}\label{E:actions}
c(a\otimes b):=(ca)\otimes b - c\otimes (ab) \ \ \text{ and } \ \ (a\otimes b)d:=a\otimes (bd)
\end{equation}
for $a,b,c\in\mathcal{C}$ and bounded Borel functions $d$.  Moreover strong locality of $\DF$ implies the right and left actions coincide  (see Theorem~2.7 of~\cite{IRT}) and the definitions (\ref{E:actions}) extend continuously to a uniformly bounded action on $\mathcal{H}$, so if $c$ is bounded and continuous and $a\in\Hil$ then
\begin{equation}\label{E:boundedactions}
\left\| ca \right\|_\mathcal{H}\leq \left\| c\right\|_{\text{sup}}\left\| a\right\|_\mathcal{H}.
\end{equation}

A derivation $\partial: \mathcal{C}\to \mathcal{H}$ can be defined by setting 
\[\partial f:= f\otimes \mathds{1}.\] 
It satisfies the Leibniz rule,
\begin{equation}\label{E:Leibniz}
\partial(fg)=f\partial g + g\partial f, \ \ f,g \in \mathcal{C},
\end{equation}
and 
\begin{equation}\label{E:normandenergy}
\left\|\partial f\right\|_\mathcal{H}^2=\mathcal{E}(f), \ \ f\in\mathcal{C}.
\end{equation}
By the latter $\partial$ extends to a linear map $\partial:\widetilde{\domDF}\to\Hil$, and
as $(\mathcal{E},\widetilde{\domDF})$ is closed in $L_2(X,\mu)$, it may be viewed as an unbounded closed operator $\partial$ from $L_2(X,\mu)$ into $\mathcal{H}$ with domain $\widetilde{\domDF}$. Let $\partial^{\ast}_\mu$ denote its adjoint, so that for $f\in\mathcal{C}$
\begin{equation*}
	\partial^{\ast}_\mu h (f) = \langle h , \partial \bar{f}\rangle_{\Hil}
	\end{equation*}
and $\partial^{\ast}_\mu:\Hil\to\mathcal{C}^{\ast}$ is a bounded linear operator into the dual $\mathcal{C}^\ast$ of the space $\mathcal{C}$ topologized by $f\mapsto \bigl(\DF(f)+\|f\|_{L_2(X,\mu)}^{2}\bigr)^{1/2}$. By standard results $\partial^\ast_\mu$ defines a densely defined unbounded operator $\partial^\ast_\mu: \mathcal{H}\to L_2(X,\mu)$. Let $(\Delta_\mu, \dom\:\Delta_\mu)$ denote the infinitesimal generator of $(\mathcal{E},\widetilde{\domDF})$. Substituting $h=\partial g$ for some $g\in\dom(\Delta_\mu)$ entails $\partial g\in\dom \partial^\ast_\mu$ and  $\Delta_\mu g=-\partial^{\ast}_\mu\partial g$. More details are in~\cite{HRT} and ~\cite{HTc}.

For $f\in\widetilde{\domDF}$ we define the support of $f\otimes\mathds{1}\in\Hil$ to be the support of the measure $\Gamma(f)$ and denote it by $\Hsupp(f\otimes\mathds{1})$.  Recall that its complement is the union of those open sets $U$ with $\Gamma(f)(U)=0$. From Remark \ref{R:Dform} (i) and the Radon property we obtain $\Gamma(f)(U)=\sup \{\int |g|^{2}d\Gamma(f): g\in\mathcal{C},\ \supp(g)\subset U,\,|g|\leq1\}$.  Recalling that $\int|g|^{2}d\Gamma(f)=\|g\partial f\|_{\Hil}^{2}$ we may then extend the definition of $\Hsupp(a)$ to all $a\in\Hil$.
\begin{definition}
The support of $a\in\Hil$ is defined by setting the complement to be
\begin{equation}\label{eqn:Hsuppdefn}
	\Hsupp(a)^{c} = \cup\{ U: U\text{ is open and } \|g a\|_{\Hil}^{2}=0 \text{ for all $g\in\mathcal{C}$ with }\supp(g)\subset U\}.
	\end{equation}
\end{definition}
This notion allows for a generalization of~\eqref{E:boundedactions}. 
\begin{lemma}
If $\Hsupp(a)$ is compact and $g$ is continuous on $\Hsupp(a)$, then
\begin{equation}\label{eqn:boundedactionsforcompactsupport}
	\|ga\|_{\Hil}\leq \|a\|_{\Hil} \sup_{x\in\Hsupp(a)}|g(x)|.
	\end{equation}
\end{lemma}
\begin{proof}
Suppose $a\in\Hil$  and $g$ are as in the statement.  For $\epsilon>0$ choose an open neighborhood $U$ of $\Hsupp(a)$ on which $|g|\leq \epsilon+\sup_{\Hsupp(a)}|g|$. Let $V$ be an open neighborhood of $\Hsupp(a)$ with $\overline{V}\subset U$ and let $\chi$ be a cut-off for $\overline{V}$ and $U$.  Observe that $\|ga\|_{\Hil}\leq \|g\chi a\|_{\Hil}+\|g(1-\chi)a\|_{\Hil}$ and the latter term is zero because of~\eqref{eqn:Hsuppdefn} and the fact that $g(1-\chi)=0$ on $V$.  Thus $\|ga\|_{\Hil}\leq (\epsilon+\sup_{\Hsupp(a)}|g|)\|a\|_{\Hil}$ by~\eqref{E:boundedactions}.
\end{proof}

Let $a\in\mathcal{H}$ be a real vector field. We may regard $a$ as a mapping $a:\mathcal{C}\to\mathcal{H}$ by $f\mapsto fa$. A magnetic operator 
(deformed differential) $\partial_a:\mathcal{C}\to\mathcal{H}$ can be defined by
\[\partial_a:=(\partial+ia), \ \ f\mapsto \partial f+ifa,\ \ f\in\mathcal{C}.\]
It is not difficult to see that 
\[\mathcal{E}^a(f,g):=\left\langle \partial_a f, \partial_a g\right\rangle_\mathcal{H}\]
defines a non-negative definite quadratic form $\mathcal{E}^a$ on $\mathcal{C}$. We have $\mathcal{E}^a(f,g)=\mathcal{E}(f,g)+\mathcal{B}(f,g)$ with
\begin{equation}\label{E:B}
\mathcal{B}(f,g)=i \langle af,\partial g\rangle_{\Hil} -i \langle \partial f,ag\rangle_{\Hil} + \langle af,ag\rangle_{\Hil},
\end{equation}
and clearly also $\mathcal{B}$ is a quadratic form on $\mathcal{C}$.

\section{Closability and self-adjointness}

The goal of this section is to prove Theorem~\ref{thm:closability}, which gives sufficient conditions for $(\DF^{a},\widetilde{\domDF})$ to be a closed extension of $(\mathcal{E}^a,\mathcal{C})$ and as a consequence, the associated magnetic Laplacian to be self-adjoint. Most of the work occurs in Lemma~\ref{lem:boundfor|fa|}.
There are two cases in the theorem.  In the first we consider a general magnetic field $a\in\Hil$ and we assume a uniform lower estimate on the measure of balls:
\begin{equation}\label{eqn:loweruniform}
	m(r):=\inf_{x\in X} \mu(B(x,r)) >0
	\end{equation}
In the second case we restrict to a compactly supported magnetic field $a\in\Hil$ and instead assume $\mu$ is a doubling measure
\begin{equation}\label{eqn:doublemeasure}
	C_{\mu}:=\sup\biggl\{ \frac{\mu(B(x,2r))}{\mu(B(x,r))}: x\in X, r> 0 \biggr\} <\infty.
\end{equation}
Note that~\eqref{eqn:doublemeasure} implies that the space is metrically doubling (see Definition~\ref{defn:rdrs}) but does not imply~\eqref{eqn:loweruniform}.

\begin{theorem}\label{thm:closability}
Suppose $(X,R,\mu)$ is a regular doubling resistance space and let $a\in\mathcal{H}$ be a real vector field.  Further assume that either (i)  $\mu$ has the  lower uniform estimate~\eqref{eqn:loweruniform} or (ii) $\mu$ is doubling as in~\eqref{eqn:doublemeasure} and $\Hsupp(a)$ is compact.  Then $(\mathcal{E}^a,\mathcal{C})$ extends to a closed quadratic form $(\mathcal{E}^a,\widetilde{\mathcal{F}})$ on $L_2(X,\mu)$ and consequently there is a unique non-positive definite self-adjoint operator $(\Delta_{\mu,a}, \dom\:\Delta_{\mu,a})$ such that 
\[\mathcal{E}^a(f,g)=-\left\langle \Delta_{\mu,a}f, g\right\rangle_{L_2(X,\mu)}\] for any $f\in \dom\:\Delta_{\mu,a}$ and $g\in\widetilde{\mathcal{F}}$.
\end{theorem}

We regard this operator $\Delta_{\mu,a}$ as the \emph{magnetic Laplacian with vector potential $a$} generated by $(\mathcal{E},\mathcal{F})$ and $\mu$. Theorem \ref{thm:closability} should be compared to the results of~\cite{HTc}, which permit the definition of a self-adjoint magnetic operator on the space $L^{2}(\Gamma)$ where $\Gamma$ is an energy-dominant measure (i.e. a measure such that all energy measures $\Gamma(f,g)$, $f,g\in\domDF$ are absolutely continuous with respect to $\Gamma$ and have bounded Radon-Nikodym derivatives).  Theorem~\ref{thm:closability} is applicable to a much wider class of measures that may be more natural on the space $X$, even though they could be singular to the energy measures, \cite{BBST, Hino03, Hino05}. In particular it is applicable to certain fractal sets $X$ which are known to support resistance forms (see~\cite{Ki01}) in the case where $\mu$ is a Hausdorff measure with respect to the resistance metric.

\begin{lemma}\label{lem:boundfor|fa|}
Let $M>0$. Under either of the assumptions of Theorem~\ref{thm:closability} the mapping $f\mapsto fa$ extends to $\widetilde{\mathcal{F}}$ and there is a constant $C_{a,M}$ depending only on $a\in\Hil$, $M$ and the properties of $X$ such that if $f\in\widetilde{\domDF}$ then
\begin{equation}\label{eqn:boundfor|fa|}
	\bigl\|fa\bigr\|_{\Hil}^2
	\leq \frac{1}{M}\DF(f) + C_{a,M}\|a\|_{\Hil}^{2} \|f\|_{L^{2}(X, \mu)}^{2}.
\end{equation}
\end{lemma}

\begin{proof}
Fix 
$0<r\leq(4MC_{d}\|a\|_{\Hil}^{2}))^{-1}$.  It is easy to see that a maximal set $\{B(x_{j},r)\}$ of disjoint balls has the property that $\cup_{j}B(x_{j},2r)\supset X$.  Moreover the metric doubling property of $X$ implies $\sum_{j}\mathds{1}_{B(x_{j},2r)}\leq C_{d}$:
if $x_{j_{1}},\dotsc,x_{j_{n}}\in B(x,2r)$ then covering by $C_{d}$ balls $B(y_{k},r)$ we see each $x_{j_{l}}$ is in some $B(y_{k},r)$ but no two can be in the same $B(y_{k},r)$ else $y_{k}\in B(x_{j_{l}},r)\cap B(x_{j_{l'}},r)$ contradicts disjointness, so $n\leq C_{d}$.

Assume first that $f\in\mathcal{C}$. We use the cover to estimate $f(x)^{2}$.  Note from~\eqref{eqn:poincare} that if $x\in B_{j}= B(x_{j},2r)$ then
\begin{equation*}
	|f(x)|^{2} \leq 2 |f(x)-f_{B_{j}}|^{2} + 2(f_{B_{j}})^{2} \leq 4\DF(f) r + 2 (f^{2})_{B_{j}}
	\end{equation*}
where the last term was estimated by Jensen's inequality.  Then for any $x$
\begin{equation*}
	|f(x)|^{2}
	\leq \sum_{j} |f(x)|^{2}\mathds{1}_{B_{j}}(x)
	\leq \sum_{j} \bigl( 4\DF(f) r + 2 (f^{2})_{B_{j}} \bigr)\mathds{1}_{B_{j}}(x)
	\leq 4C_{d} \DF(f) r + \sum_{j} 2 (f^{2})_{B_{j}} \mathds{1}_{B_{j}}(x).
	\end{equation*}
If $\mu$ has the lower uniformity property~\eqref{eqn:loweruniform} then the last term is bounded by $2C_{d}\|f\|_{L_2(X,\mu)}^{2}(m(r))^{-1}$.  If, instead, $\mu$ is doubling and $a$ has compact support then let $B(x_{0},\rho)$ contain $\Hsupp(a)$, take $k$ so $2^{k}r>2\rho$ and for $x\in B(x_{0},\rho)$  iterate the doubling estimate~\eqref{eqn:doublemeasure} to see $\mu(B(x,r))\geq C_{\mu}^{-k}\mu(B(x_{0},\rho))$. Then the last term in the above equation is bounded by $2C_{d} C_{\mu}^{k} \|f\|_{L_2(X,\mu)}^{2}\mu(B(x_{0},\rho))^{-1}$ if $x\in B(x_{0},\rho)$.  Therefore we have a value $C_{a,M}$ depending on $a, M$ and on the properties of $X$ such that, after using our choice of $r$ to simplify the first term,
\begin{equation}\label{eqn:LinftybdfromDFandL2}
	|f(x)|^{2} \leq \frac{\DF(f)}{M\|a\|_{\Hil}^{2}} + C_{a,M} \|f\|_{L^{2}(X,\mu)}^{2}
	\end{equation}
holds for all $x\in X$ under the lower uniformity assumption, or for all $x$ in a neighborhood of the compact set $\Hsupp(a)$ under the doubling assumption.  The proof of (\ref{eqn:boundfor|fa|}) for $f\in\mathcal{C}$ is now completed using~\eqref{E:boundedactions} in the former case and~\eqref{eqn:boundedactionsforcompactsupport} in the latter case. For general $f\in\widetilde{\mathcal{F}}$ (\ref{eqn:boundfor|fa|}) follows by approximation according to Remark \ref{R:Dform} (ii).
\end{proof}

\begin{proof}[\protect{Proof of Theorem~\ref{thm:closability}}]
Note first that by Lemma \ref{lem:boundfor|fa|} the quadratic form
\begin{equation}\label{E:temporaryB}
\mathcal{B}(f)=\mathcal{E}^a(f)-\mathcal{E}(f)=2\mathfrak{R}(i\left\langle \partial f, fa\right\rangle_\mathcal{H})+\left\|fa\right\|_\mathcal{H}^2
\end{equation}
is defined for all $f\in\widetilde{\mathcal{F}}$. Applying Cauchy-Schwarz to the cross-term in (\ref{E:temporaryB}), we obtain
\[|\mathcal{B}(f)|\leq \frac14\mathcal{E}(f)+5\left\|fa\right\|_\mathcal{H}^2,\]
which by~\eqref{eqn:boundfor|fa|} yields 
\begin{equation}\label{E:formbound}
|\mathcal{B}(f)|\leq \varepsilon\:\mathcal{E}(f)+C\left\|f\right\|_{L_2(X,\mu)}^2,\ \ f\in\widetilde{\mathcal{F}},
\end{equation}
with positive constants $C:=5C_{a,M}\left\|a\right\|_{\mathcal{H}}^2$ and $\varepsilon:=(4^{-1}+M^{-1}+4M^{-1})<1$, provided $M>20/3$. Now the result follows from the classical KLMN theorem (\cite[Theorem X.17]{RS}).
\end{proof}

Recall that in the definition of $\partial_{a}$ we treat $a\in\Hil$  as a linear  operator $a:\mathcal{C}\to\Hil$ via $f\mapsto fa$ with the bound $\|fa\|_{\Hil}\leq \|f\|_{\sup} \|a\|_{\Hil}$.  From the proof of Lemma~\ref{lem:boundfor|fa|} we see that either $\|f\|_{\sup}^{2}\leq 4C_{d}\DF(f)r + 2C_{d}(m(r))^{-1}\|f\|_{L_2(X,\mu)}^{2}$ or $\|f\|_{\sup}^{2}\leq 4C_{d}\DF(f)r + 2C_{d}C_{\mu}^{k}\mu(B(x_{0},\rho))^{-1}\|f\|_{L_2(X,\mu)}^{2}$.  In either case we can optimize over $r$ to obtain 
\[\|f\|_{\sup}\leq C(\DF(f)+\|f\|_{L_2(X,\mu)}^{2})^{1/2}.\]  
This allows us to define an adjoint operator $a^{\ast}_{\mu}:\Hil\to\mathcal{C}^{\ast}$ simply as $h\mapsto a^{\ast}_{\mu}h$, where for $f\in\mathcal{C}$ we set 
\begin{equation*}
a^{\ast}_{\mu}h(f)=\langle h , \bar{f}a  \rangle_{\Hil}.
\end{equation*}
Both the fact that this defines $a^{\ast}h$ as an element of $\mathcal{C}^{\ast}$ and the boundedness of the map $a^{\ast}_{\mu}$ follow from
\begin{equation*}
	|\langle fa, \bar{h}\rangle_{\Hil}|
	\leq \|f\|_{\sup}\|a\|_{\Hil}\|h\|_{\Hil}
	\leq C(\DF(f)+\|f\|_{L_2(X,\mu)}^{2})^{1/2} \|a\|_{\Hil}\|h\|_{\Hil}
	\end{equation*}
which was obtained using Cauchy-Schwarz, \eqref{E:boundedactions} and our bound for $\|f\|_{\sup}$.

It is then natural to define an adjoint of $\partial_{a}$.  We set $\partial_{\mu,a}^{\ast}=(\partial+ia)^{\ast}_{\mu}$, so that
\begin{equation*}
	\DF^{a}(f,g) = \langle (\partial+ia)f,(\partial+ia)g\rangle_{\Hil}
	=  (\partial^{\ast}_{\mu} - ia^{\ast}_{\mu}) \bigl((\partial+ia)f\bigr)(\bar{g}) 
	=\partial_{\mu,a}^{\ast} \partial_{a} f (\bar{g})
	\end{equation*}
for any $f,g\in\mathcal{C}$ and therefore 
\[\Delta_{\mu,a}f=-\partial_{\mu,a}^{\ast} \partial_{a}f, \ \ f\in\mathcal{C},\]
seen as an identity in $\mathcal{C}^\ast$. For $f\in\dom\Delta_{\mu,a}$ it can be interpreted as an identity in $L_2(X,\mu)$. We record a simple fact about the spectra of $\Delta_\mu$ and $\Delta_{a,\mu}$.

\begin{theorem}\label{thm:compactness}
Let the hypotheses of Theorem~\ref{thm:closability} be in force. If $\Delta_{\mu}$ has pure point spectrum, then also $\Delta_{\mu,a}$ has compact resolvent and therefore pure point spectrum with eigenvalues $0\geq\nu_{1}\geq\nu_{2}\geq \dots$ accumulating only at $-\infty$.
\end{theorem}

Under mild conditions resistance forms on p.c.f. self-similar fractals always lead to Laplacians $\Delta_\mu$ with with pure point spectrum (or, equivalently, with compact resolvent), see for example~\cite[Lemma~3.4.5]{Ki01}. This is essentially due to the fact that for any element of the dense subspace $\mathcal{C}$ there are finite-dimensional approximations given by the resistance condition. Other examples of Laplacians with pure point spectrum arise from resistance forms on (generalized) Sierpinski carpets, \cite{BB, BBKT}.

\begin{proof} Since $\Delta_\mu$ has pure point spectrum, $\widetilde{\domDF}$ is compactly embedded into $L^{2}(\mu)$, see~\cite{Ki01} Theorem~B.1.13. This remains true if $\widetilde{\mathcal{F}}$ is normed by $(\mathcal{E}^a+\|\cdot\|_{L_2(X,\mu)}^2)^{1/2}$. Thus $\Delta_{\mu,a}$ has pure point spectrum and compact resolvent (see again~\cite{Ki01} Theorem~B.1.13). 
\end{proof}


\section{Locality, local exactness and gauge invariance}\label{sec:gauge}
Using the energy measure representation of the norm in $\Hil$ we determine that the form $\DF^{a}$ from Theorem~\ref{thm:closability} is local, meaning that  $\DF^{a}(f,g)=0$ for all $f,g$ with disjoint supports; this property does not depend on the measure $\mu$. 

Throughout this section we assume that $(X,R)$ is compact, what implies $\mathcal{C}=\mathcal{F}$.

\begin{lemma}
If $a,b,c,d\in\mathcal{F}$ and $\supp(a)\cap\supp(c)\cap\supp(bd)=\emptyset$ then $\langle a\otimes b,c\otimes d\rangle_{\Hil}=0$.
\end{lemma}
\begin{proof}
Write out the expression in terms of the energy measure
\begin{align*}
	2\langle a\otimes b,c\otimes d\rangle_{\Hil}
	&= 2\int_{X} bd \, d\Gamma(a,c)\\
	&= \DF(bda,c)+\DF(a,bdc)-\DF(ac,bd).
	\end{align*}
which is zero by locality of $\DF$ and the support assumption. 
\end{proof}

\begin{theorem}
Fix $a\in\Hil$.  If the conditions of Theorem~\ref{thm:closability} are satisfied and $f,g\in\domDF$ have disjoint supports then $\DF^{a}(f,g)=0$.
\end{theorem}
\begin{proof}
Recall $a\in\Hil$ can be approximated by linear combinations of the form $\sum_{j=1}^{n}a_{j}\otimes b_{j}$ with $a_{j},b_{j}\in\mathcal{C}$.  Then $af$ is approximated by $\sum_{j=1}^{n}a_{j}\otimes b_{j}f$, see~\eqref{E:boundedactions}.  Similarly $ag$ is approximated by $\sum_{j=1}^{n}a_{j}\otimes b_{j}g$.  Applying the previous lemma, for all $j$
\begin{gather*}
	\langle \sum_{j=1}^{n}a_{j}\otimes b_{j}f,  g\otimes \mathds{1} \rangle_{\Hil}=0\\
	\langle f\otimes \mathds{1}, \sum_{j=1}^{n}a_{j}\otimes b_{j}g\rangle_{\Hil}=0\\
	\langle \sum_{j=1}^{n}a_{j}\otimes b_{j}f,\sum_{j=1}^{n}a_{j}\otimes b_{j}g \rangle_{\Hil} =0
	\end{gather*}
and by taking limits using~\eqref{E:boundedactions} we see that $B(f,g)$ as considered in (\ref{E:B}) is zero,
\[\mathcal{B}(f,g)=i \langle af,\partial g\rangle_{\Hil} -i \langle \partial f,ag\rangle_{\Hil} + \langle af,ag\rangle_{\Hil}=0.\]
Since $\mathcal{E}(f,g)=0$ by locality, the result follows. 
\end{proof}

Our next goal is to show that modifying the magnetic field by adding a gradient is equivalent to conjugating the associated magnetic form by an exponential.  This property is called gauge invariance (and the gradient is referred to as a gauge field).  We need a trivial lemma.

\begin{lemma}\label{lem:zeroderivimpliesconst}
If $f\in\domDF$ and $\|\partial f\|_{\Hil}=0$ then $f$ is constant.
\end{lemma}
\begin{proof}
Recall $\|\partial f\|_{\Hil}^{2}=\DF(f)$, so this follows from (RF\ref{RF1}).
\end{proof}


We also recall the following result of LeJan regarding strong local forms.
\begin{theorem}[Theorem 3.2.2. in~\protect{\cite{FOT94}}]\label{thm:LeJan}
If $\Phi\in C^{1}(\mathbb{R}^{m})$  with $\Phi(0)=0$ and $u=(u_{1},\dotsc,u_{n})\in\domDF^{n}$ with all $u_i$ bounded then $\Phi(u)\in\domDF$ is bounded and for all bounded $f\in\domDF$
\begin{equation}\label{eqn:LeJan}
	d\Gamma(\Phi(u),f) = \sum_{j=1}^{n} \frac{\partial\Phi}{\partial x_{j}}(u) \, d\Gamma(u_{j},f)
	\end{equation}
Without condition $\Phi(0)=0$ the function $\Phi(u)$ is a member of $\domDF_{\text{loc}}$, the space of functions which are locally in $\domDF$ in the sense that on any open set with compact closure they agree with a function from $\domDF$. Formula (\ref{eqn:LeJan}) remains valid in this case. The assumption that the $u_{j}$ are bounded can be removed if all partial derivatives of $\Phi$ are bounded. 
\end{theorem}

In particular we  see that for all $f\in\domDF$ the function $e^{if}-1$ exists in $\domDF$ and satisfies the above.  From the theorem the function $e^{if}$ is only locally in $\domDF$, but in the Hilbert space of $\domDF$ modulo constants it is simply $e^{if}-1$.  Notice also that $\DF(e^{if})=\DF(f)$.

\begin{remark}
Theorem~\ref{thm:LeJan} immediately implies that $\partial \Phi(u)$ is a member of $\mathcal{H}$ (or locally a member) and 
\begin{equation}\label{eqn:LeJan2}
\partial \Phi(u)=\sum_{j=1}^n\frac{\partial\Phi}{\partial x_{j}}(u)\, \partial u_j
\end{equation}
holds (globally or locally, respectively).
\end{remark}

\begin{theorem}\label{thm:gaugeinvariance}
Let $\lambda\in\domDF$.  The solution set of $(\partial+i\partial\lambda)f=0$ consists of the constant multiples of $e^{-i\lambda}$.  Moreover we have gauge invariance: $\DF^{a}(e^{i\lambda}f)=\DF^{a+\lambda}(f)$.  In particular
\begin{equation*}
	\Delta_{\mu,a+\partial\lambda} = e^{-i\lambda} \Delta_{\mu,a} e^{i\lambda}.
	\end{equation*}
so that $\Delta_{\mu,a}$ and $\Delta_{\mu,a+\partial\lambda}$ have the same spectrum and their domains are related by multiplication by $e^{i\lambda}$.
\end{theorem}
\begin{proof}
Theorem~\ref{thm:LeJan} implies $e^{-i\lambda}$ exists and is locally in $\domDF$.  For any constant $c$ we compute by~\eqref{eqn:LeJan} 
\begin{align*}
	\big\| \partial (ce^{-i\lambda}) + ie^{-i\lambda}\partial(c\lambda) \bigr\|_{\Hil}^{2}
	&= \big\| ce^{-i\lambda}\otimes\mathds{1}   \bigr\|_{\Hil}^{2} + \big\|  ie^{-i\lambda}(c\lambda\otimes\mathds{1}) \bigr\|_{\Hil}^{2}	+2 \bigl\langle ce^{-i\lambda} \otimes\mathds{1} , ie^{-i\lambda}(c\lambda\otimes\mathds{1})  \bigr\rangle_{\Hil}\\
	&= \int d\Gamma(ce^{-i\lambda}) + \int (ie^{-i\lambda})^{2}\,d\Gamma(c\lambda) + 2\int (ie^{-i\lambda}) d\Gamma(ce^{-i\lambda},c\lambda)\\
	&=\int  (-ice^{-i\lambda})^{2} \,d\Gamma(\lambda) + \int c^{2}(ie^{-i\lambda})^{2}\,d\Gamma(\lambda) + 2\int  (ie^{-i\lambda})(-ice^{-i\lambda}) c\,d\Gamma(\lambda)\\
	&=0
	\end{align*}
so that $ce^{-i\lambda}$ is a solution to the equation.  Conversely if $f$ is any solution then a similar computation gives $\partial(e^{i\lambda}f)=0$ so by Lemma~\ref{lem:zeroderivimpliesconst} $f$ is a constant multiple of $e^{-i\lambda}$.

For the gauge invariance we can use~\eqref{eqn:LeJan2} to compute
\begin{align*}
	\DF^{a}(e^{i\lambda}f)
	&=\langle (\partial+ia)e^{i\lambda}f,  (\partial+ia)e^{i\lambda}f \rangle_{\Hil}\\
	&= \langle e^{i\lambda} (\partial + i(a+\partial\lambda))f, e^{i\lambda} (\partial + i(a+\partial\lambda))f \rangle_{\Hil}\\
	&= \langle  (\partial +i(a+\partial\lambda))f,  (\partial + i(a+\partial\lambda))f \rangle_{\Hil}\\
	&=\DF^{a+\partial\lambda}(f),
	\end{align*}
from which the asserted results are immediate.
\end{proof}

\begin{corollary}\label{cor:canconjugateexactforms}
If $a$ is exact, meaning $a=\partial\lambda$ for some $\lambda\in\domDF$  then $\DF^{a}(e^{-i\lambda}f) = \DF(f)$ and $\Delta_{\mu,a}=e^{i\lambda}\Delta_{\mu}e^{-i\lambda}$.
\end{corollary}

It should be noted that in this circumstance the standard properties of resistance forms carry over to $\DF^{a}$ via the conjugation.  For example, for any non-empty $Y\subset X$ there is a Green function $g_{Y}:X\times X\to\mathbb{R}$ such that $\DF(g_{Y}(x,\cdot),f(\cdot))=f(x)$ for all $f\in\domDF$ that vanish on $Y$ (see~\cite{Ki12}).  It is readily verified that then $\DF^{a}(e^{-i\lambda}g_{Y}(x,\cdot),e^{-i\lambda}f(\cdot))=f(x)$ for the same $f$, which is the same as saying we can solve $\Delta_{\mu,a}u=f$ on $X\setminus Y$ and $u=0$ on $Y$ using conjugation with $e^{i\lambda}$ and integration against $g_{Y}$.

Corollary~\ref{cor:canconjugateexactforms} also has consequences for studying $\DF^{a}$.  Recall that the exact forms span a subspace $\{\partial f:f\in\domDF\}$ of $\Hil$. By (RF\ref{RF2}) it is closed. Any $a\in\Hil$ may then be written as the sum of an exact form and a form orthogonal to the exact forms. We may use Corollary~\ref{cor:canconjugateexactforms} to conjugate away the exact part of $a$, so it suffices to study $\DF^{a}$ when $a\in\Hil$ is orthogonal to the exact forms. Such forms $a$ are usually referred to as \emph{Coulomb gauges}.

We may improve Theorem~\ref{thm:gaugeinvariance} to fields that are only locally exact provided there is a non-trivial $f\in\domDF$ solving $\DF^{a}(f)=0$, but we need certain additional assumptions on $X$.

\begin{assumption}\label{assum:localexactetc}\ 
\begin{enumerate}
\item $X$ is locally connected, and
\item If $U\subset X$ is open and connected and $f\in\domDF$ satisfies $(\partial f)\mathds{1}_{U}=0$ in $\Hil$ then $f$ is constant on $U$.
\end{enumerate}
\end{assumption}
\begin{remark}
If $X$ has a finitely ramified cell structure as in~\cite{T08} then the latter two hypotheses are both true.  In particular these hold for the class of post-critically finite self-similar fractals of Kigami~\cite{Ki01}.
\end{remark}

\begin{definition}\label{defn:locexact}
We say $a\in\Hil$ is locally exact if there is an open cover $\cup_{j}U_{j}$ and functions $\lambda_{j}\in\domDF$ such that $a\mathds{1}_{U_{j}}=(\partial\lambda_{j})\mathds{1}_{U_{j}}$ for all $j$.
\end{definition}
It is worth noting that in many cases of interest locally exact forms are not typical.  For example, on fractal gaskets and carpets there is non-trivial topology at all locations and scales, see \cite{HT} and \cite{IRT}.

\begin{theorem}\label{thm:DFa=0}
Let $(X,\DF)$ satisfy Assumption~\ref{assum:localexactetc}.  Fix real-valued $a,b\in\Hil$ and suppose the hypotheses of Theorem~\ref{thm:closability} are satisfied for both $a$ and $b$.  Further suppose $a$ is locally exact.  If there is a non-zero $f\in\domDF$ such that $\DF^{a}(f)=0$ then $|f|$ is constant on $X$ and the set $\{g:\DF^{a}(g)=0\}$ consists of the constant multiples of $f$.  Taking $f_{0}$ with $\DF^{a}(f_{0})=0$ and $|f_{0}|=1$ we have gauge invariance $\DF^{a+b}(f_{0}f) = \DF^{b}(f)$.  In particular $\Delta_{\mu,a+b}= f_{0}^{-1} \Delta_{\mu,b} f_{0}$, so these operators have the same spectrum and their domains are related by multiplication by $f_{0}$.
\end{theorem}
\begin{proof}
Local exactness and local connectedness provide a cover $\{U_{j}\}$ of $X$  by connected open sets and corresponding $\lambda_{j}\in\domDF$ so that $a\mathds{1}_{U_{j}}=(\partial\lambda_{j})\mathds{1}_{U_{j}}$ for each $j$. Note that the $\lambda_{j}$ are real-valued.  If $\|(\partial+ia)f\|^{2}_{\Hil}=\DF^{a}(f)=0$ then
\begin{equation*}
	\bigl( \partial (fe^{i\lambda_{j}}) \bigr) \mathds{1}_{U_{j}}
	= \bigl(  (\partial f) e^{i\lambda_{j}} + i(\partial\lambda_{j})e^{\lambda_{j}}f  \bigr)\mathds{1}_{U_{j}} 
	=\bigl( (\partial+ ia)f\bigr) \mathds{1}_{U_{j}}
	=0.
	\end{equation*}
By the third point in Assumption~\ref{assum:localexactetc} we conclude $fe^{i\lambda}$ is constant on $U_{j}$, so $f=c_{j}e^{-i\lambda_{j}}$ for some constant $c_{j}$.  Moreover if $U_{j}\cap U_{k}$ is non-empty we must have $c_{j}e^{-i\lambda_{j}}=c_{k}e^{-i\lambda_{k}}$.  Then $|f|=|c_{j}|=|c_{k}|$, and since any pair of cells are connected by a chain of sets $U_{l}$ with non-empty intersections we see $|f|$ is constant.  There is no loss of generality in taking this constant to be $1$, and the phase in $c_{j}=e^{i\theta_{j}}$ may be absorbed by replacing $\lambda_{j}$ with $\lambda_{j}-\theta_{j}$, for some real constants $\theta_{j}$.  Thus $\DF^{a}(f)=0$ if and only if there is a choice of $\lambda_{j}\in\domDF$ with $a\mathds{1}_{U_{j}}=(\partial\lambda_{j})\mathds{1}_{U_{j}}$ and $f\mathds{1}_{U_{j}}=c e^{-i\lambda_{j}}$ for all $j$. The latter is equivalent to $\lambda_{j}-\lambda_{k}\in 2\pi\mathbb{Z}$ when $U_{j}\cap U_{k}\neq\emptyset$. We let $f_{0}$ be the case $c=1$.

For the gauge invariance,
\begin{align*}
	\DF^{a+b}(f_{0}f)
	&= \| (\partial+ia+ib)(f_{0}f) \|^{2}_{\Hil}\\
	&= \| f (\partial + ia)f_{0} + f_{0} (\partial+ ib) f \|^{2}_{\Hil}\\
	&= \|f_{0} (\partial+ ib) f \|^{2}_{\Hil}
	= \| (\partial+ ib) f \|^{2}_{\Hil} = \DF^{b}(f)
	\end{align*}
because $(\partial+ia)f_{0}=0$ in $\Hil$ and $|f_{0}|=1$ everywhere.
\end{proof}

\begin{remark}
Of course the situation where there is non-trivial $f$ so $\DF^{a}(f)=0$ is also that where we can define $e^{ia}$ to be equal $e^{i\lambda_{j}}$ on each $U_{j}$, because this definition is legitimate if and only if on those $U_{j}\cap U_{k}\neq\emptyset$ one has $\lambda_{j}-\lambda_{k}$ equal to an integer multiple of $2\pi$.
\end{remark}

\section{Examples}\label{sec:examples}
The only classical Dirichlet spaces that are also resistance spaces are one-dimensional, but interest in resistance spaces has developed substantially since it was realized that many classes of fractal sets also admit resistance forms.  Our theory is applicable to these examples provided the measure is sufficiently well-behaved, and in practice the latter limitation is minor because the natural choices of measure have the necessary properties.

\begin{example}[Circle with fractal mass]
The unit circle with the form $\DF(u)=\int |u'|^{2}$ is a regular doubling resistance space and the form is strongly local.  In the case that the measure $\mu$ is Lebesgue measure we merely recover the usual  theory of magnetic fields on the circle, however our approach is also applicable to a doubling measure $\mu$ that is singular to the Lebesgue measure, because we then satisfy the hypotheses of Theorem~4.1.  One natural class of such measures is obtained by viewing functions on the circle as periodic functions on the unit interval and the latter as a post-critically finite self-similar space under a finite collection of similarities mapping the interval to a union of subintervals. For most choices of  self-similar structure the corresponding Bernoulli measure is both doubling and singular to the Lebesgue measure.  In this sense our results apply to magnetic fields on the circle with mass given by a fractal Bernoulli measure.
\end{example}

\begin{example}[Postcritically finite self-similar sets with Bernoulli measures]
The circle with a self-similar measure is a very special case in the general class of postcritically finite self-similar sets, see \cite{Ki01}.  There are much larger classes for which the existence of a regular Dirichlet form is known (e.g. the nested fractals defined by Lindstr\"{o}m and their generalizations~\cite{Lind,Kumagai}).  These forms are generally self-similar, so the metric doubling condition reduces to comparability of a finite set of resistance scaling coefficients.  The spaces are compact so our theory is applicable once the measures are doubling.  When a  self-similar measure is used the doubling condition again amounts to verifying comparability of a finite set of coefficients.  The canonical example of such a set is the Sierpinski Gasket.
\end{example}

\begin{example}[Sierpinski Carpets]
There are considerable technical difficulties in constructing Dirichlet forms on self-similar sets with infinite ramification, but this has been done successfully for a class of highly symmetric Sierpinski Carpets~\cite{BB}.  In many cases the construction gives a unique resistance form,~\cite{BB, KusuokaZhou}. If we consider a Bernoulli self-similar measure then our results apply.
\end{example}

\begin{example}[Fractafolds]
Strichartz~\cite{Strichartzfractafolds} has proposed a notion of fractafolds based on self-similar fractals and has studied those based on the Sierpinski Gasket.  In particular this allows us to consider non-compact spaces which are locally like the fractals in the previous examples.  The simplest type considered in~\cite{Strichartzfractafolds} are based on post-critically finite self-similar fractals and have a cell-structure, meaning that the fractafold is a union of cells that are copies of the underlying fractal, perhaps with some rescaling of the resistance or measure.  It is not hard to see that in the case where the measure and resistance scalings for a cell are within bounds independent of the cell then the measure estimate~\eqref{eqn:loweruniform} holds and our theory applies.  Alternatively one could construct a fractafold for which measure doubling holds but~\eqref{eqn:loweruniform} fails, in which case our theory applies to compactly supported magnetic fields.
\end{example}

\end{document}